\documentclass{article}

\usepackage{amsfonts,amssymb,amsmath,amsthm,color,graphicx}
\usepackage[height=7.8in,width=5.8in]{geometry}
\theoremstyle{definition}
\newtheorem{theorem}{Theorem}[section]
\newtheorem{definition}[theorem]{Definition}

\begin{document}

\title{Guessing Games} 
\author{Anthony Mendes and Kent E. Morrison}
\date{}
\maketitle

\begin{abstract} 
In a guessing game, players guess the value of a random real number selected using some probability 
density function.  The winner may be determined in various ways; for example, a winner can be a  
player whose guess is closest in magnitude to the target or a winner can be a player coming 
closest without guessing higher than the target.  We study optimal strategies for players in these 
games and determine some of them for two, three, and four players.
\end{abstract}

\section{Introduction and Background.}

There are situations in everyday life in which people try to outguess one another.  Visitors to the 
county fair can be asked to guess the number of jellybeans in a jar.  In a sealed-bid auction, it
might be desirable to bid closest to the value of a some object without guessing too high.  
Contestants on the American television program 
``The Price Is Right'' try to make better guesses on the price of household goods than their opponents.  
Assuming rational and intelligent guessers, what does the mathematics say should a person do? 

A collection of $n$ players guess the value of $r$, a random number selected using 
some cumulative distribution function $G$ known to all.  How should each player guess if the 
winning guess is the closest without going higher than $r$?  How should each player guess if
the winning guess is simply the closest guess to $r$?  We provide techniques to find solutions to 
these games and exhibit explicit solutions for small values of $n$.  

The goal is to find a single optimal strategy $S$ such that a Nash equilibrium results when 
each player employs $S$.  In other words, if player $i$ knew that all of his opponents 
were using strategy $S$, then the expected payoff to player $i$ is maximized when player $i$ uses 
strategy $S$ as well.  For our guessing games, a strategy is a cumulative distribution function $F$
such that $F(x)$ is the probability of guessing a real number less than or equal to $x$.

Under the assumption that the random number $r$ is a continuous random variable, the general problem can 
be reduced to the case that $r$ is selected uniformly in $[0,1]$.  The cumulative distribution function $G$ 
is strictly increasing on the range of $r$ and so $G^{-1}$ exists.  If $F(x)$ is an optimal 
strategy when $r$ is selected uniformly from $[0,1]$, then $F(G(x))$ is an optimal strategy for the more 
general case.  This is because
\begin{align*}
F(G(x)) 
& = \text{the probability that a guess for $r$ is $\leq G(x)$} \\
& = \text{the probability that a guess for $G^{-1}(r)$ is $\leq x$},
\end{align*}
providing a cumulative distribution function.  So we restrict attention to when 
$r$ is selected uniformly in $[0,1]$.  We will assume that the players choose numbers 
in $[0,1]$ since there is nothing to be gained from other choices.

In section \ref{pir} we analyze the game in which the winner makes the guess closest to, but not larger than, 
$r$. In section \ref{closest} we analyze the game in which the winner is simply the player whose guess is 
closest to $r$.  The paper ends with some remarks about numerical approximations in section \ref{approx}.  This 
introductory section concludes by clarifying what we mean by a solution in an $n$ person zero-sum game 
and by proving two theorems about solutions.

Representing the simultaneous choices of the $n$ players as a vector $\mathbf{x} = (x_1,x_2,\ldots,x_n)$ in 
$[0,1]^n$, we denote the payoff to player $i$ by $A_i(\mathbf{x})$. This is the expected return to player $i$ 
when the target number $r$ is distributed uniformly in $[0,1]$. If each player pays one unit to play and the 
winner takes all, then the game is \emph{zero-sum}, i.e., $\sum_i A_i(\mathbf{x}) =0$ for all $\mathbf{x}$. 
(With ties the pot is split equally.)

Since the set of pure strategies for each player is the unit interval $[0,1]$, a mixed strategy is a 
probability measure on the interval or, equivalently, a cumulative distribution function. Suppose that player 
$i$ uses the mixed strategy given by the cumulative distribution function $F_i$, while the other players guess 
$x_j$ for $j \neq i$. Then the expected payoff to player $i$ is the Riemann-Stieltjes integral
\begin{equation*}
\int A_i(\mathbf{x}) \, dF_i(x_i),
\end{equation*}
which, when $F'(x_i)$ exists and is bounded, is the same as the ordinary integral
\begin{equation*}
\int A_i(\mathbf{x}) F_i'(x_i) \, dx_i.
\end{equation*}
It will turn out that the optimal mixed strategies that we find are actually differentiable, but it is 
convenient to allow the larger class of cumulative distribution functions at the outset.  

\begin{definition}
Let $A_i$ be the payoff function for player $i$ in a zero-sum game with $n$ players. A \emph{solution} to the 
game is the data consisting of cumulative distribution functions $F_1,\dots,F_n$ and real numbers 
$v_1,\dots,v_n$ such that $\sum_i v_i = 0 $ and for each $i=1,\ldots,n$ and for all $x_j \in [0,1]$ with 
$j \neq i$, the inequality
\begin{equation*} 
v_i \leq \int  A_i(\mathbf{x}) \, dF_i(x_i)
\end{equation*}
holds.  We refer to the $F_i$ as \emph{optimal strategies}.
\end{definition}

By definition, using an optimal strategy $F_i$ guarantees player $i$ an expected payoff of at least $v_i$ 
regardless of what the other players do.  It follows that 
\begin{equation*}
v_i \leq 
\int \cdots \iiint  \cdots \int A_i(\mathbf{x}) \, dG_1(x_1) \cdots dG_{i-1}(x_{i-1}) dF_i(x_i) 
dG_{i+1}(x_{i+1}) \cdots dG_{n}(x_{n}) 
\end{equation*}
for any cumulative distribution functions $G_1,\dots,G_{i-1},G_{i+1},\dots,G_{n}$.
The condition that $\sum_i v_i = 0$, combined with the fact that we are considering a zero-sum game, implies 
that $v_i$ is the maximal expected payoff that player $i$ can ensure for himself.   This fact is implied by 
Theorem \ref{equals} below.

\begin{theorem}
\label{equals}
If $F_1,\dots,F_n$ together with $v_1,\dots,v_n$ is a solution for a zero-sum game with payoffs 
given by $A_i$, then
\begin{equation*}
v_i = \int \cdots \int  A_i(\mathbf{x}) \, dF_1(x_1)  \cdots dF_n(x_n).
\end{equation*}
This iterated integral gives the expected payoff to player $i$ when each player $j$ is using the 
cumulative distribution function $F_j$ to select his real number.  So Theorem \ref{equals}  
says that the expected payoff to player $i$ is $v_i$ when each player is using an optimal strategy.
\end{theorem}

\begin{proof} 
We have 
\begin{align*}
v_i & \leq \int \cdots \int  A_i(\mathbf{x}) \, dF_1(x_1)  \cdots dF_n(x_n)  \\
& = \int \cdots \int  \left( -\sum_{j \ne i}  A_j(\mathbf{x}) \right) \, dF_1(x_1) \cdots dF_n(x_n)\\
& = - \sum_{j \ne i} \int \cdots \int  A_j(\mathbf{x}) \, dF_1(x_1)  \cdots dF_n(x_n)  \\
& \leq  - \sum_{j \ne i} v_j \\
& = v_i
\end{align*}
and so we must have equalities instead of inequalities throughout. 
\end{proof}

Theorem \ref{eq} below says that if every player except player $i$ uses an optimal cumulative 
distribution function, then player $i$ will earn an expected payoff of $v_i$.  Theorem \ref{eq} may 
be considered a continuous analogue of the equilibrium theorem for games.

\begin{theorem}
\label{eq}
If $F_1,\dots,F_n$ together with $v_1,\dots,v_n$ is a solution for a zero-sum game with payoffs 
given by $A_i$, then
\begin{equation*}
v_i = 
\int \cdots \int \int \cdots \int A_i(\mathbf{x}) \, dF_1(x_1) \cdots 
dF_{i-1}(x_{i-1}) dF_{i+1} (x_{i+1}) \cdots dF_n(x_n).
\end{equation*}
There are $n-1$ integrals here; the missing integral corresponds to $dF_i$.   
\end{theorem}

\begin{proof}
Since $v_i$ is the 
maximum expected payoff that player $i$ can ever achieve when his opponents are using optimal 
strategies, we have
\begin{equation*}
v_i \geq
\int \cdots \int \int \cdots \int A_i(\mathbf{x}) \, dF_1(x_1) \cdots 
dF_{i-1}(x_{i-1}) dF_{i+1} (x_{i+1}) \cdots dF_n(x_n).
\end{equation*}
Now suppose that equality does not hold so that
\begin{equation*}
v_i  >
\int \cdots \int \int \cdots \int A_i(\mathbf{x}) \, dF_1(x_1) \cdots 
dF_{i-1}(x_{i-1}) dF_{i+1} (x_{i+1}) \cdots dF_n(x_n).
\end{equation*}
Integrate both sides of the inequality with respect to $x_i$ to get
\begin{multline*}
\int (v_i) \, dF_i(x_i) 
>\int \bigg(  \int \cdots \int \int \cdots \int A_i(\mathbf{x}) \, dF_1(x_1) \\ \cdots 
dF_{i-1}(x_{i-1}) dF_{i+1} (x_{i+1}) \cdots dF_n(x_n)   \bigg)  \, dF_i(x_i) . 
\end{multline*}
The left side is $v_i$ and the right side is
\[ \int \cdots \int  A_i(\mathbf{x}) \, dF_1(x_1)  \cdots dF_n(x_n) , \]
which is $v_i$ by Theorem 1.2. Thus we have $v_i > v_i$, a contradiction, and so equality must hold.
\end{proof}

\section{The Price Is Right.}
\label{pir} 

In this section we consider the following guessing game.  A random number $r$ is selected uniformly 
in $[0,1]$.  A total of $n$ people guess $r$.  A winning guess is a guess closest to, but still 
smaller than, $r$.  This game is very similar to guessing game played on the daytime television 
program \emph{The Price Is Right}. The main difference is that guesses on the television program are
given sequentially so that the $i^{\textrm{th}}$ guesser has the privilege of knowing the previous 
$i-1$ guesses.

If $k \geq 1$ players out of $n$ select the same winning number, then we define the payoffs to be 
$n/k - 1$ to each winner and $-1$ to everyone else.  If all players guesses are larger than $r$, the
payoff to all players is $0$.  This zero-sum payoff scheme is the same as if each player antes $1$ to
play and winners share the pot.
 
We now provide a solution to the two person game using a technique outlined in \cite{owen} to solve 
certain continuous two person games.  

The two person game is symmetric (meaning that the game is the same for either player), so the 
expected payoff for both players should be $0$.  Additionally, there should be one strategy, given 
by a cumulative distribution function $F$, which is optimal for both players.  In order to find 
$F$ in the proof of Theorem \ref{2} below, we assume $F$ is constant except on some interval $(0,u)$ 
and that $F'(x)$ exists and is positive on $(0,u)$.   These natural assumptions allow us to replace 
Riemann-Stieltjes integrals with Riemann integrals in our calculations. 

\begin{theorem} 
\label{2}
An optimal strategy for the two person ``Price Is Right'' guessing game is the cumulative 
distribution function defined by $F(x) = 1/\sqrt{1-x} - 1$ for $x \in (0,3/4)$. 
\end{theorem}

\begin{proof} 
If the first player selects $x$ and the second selects $y$ for $x,y \in [0,u]$, the expected payoff 
$A(x,y)$ to the first player is
\begin{equation*}
A(x,y) 
= 
\begin{cases}
2 y - x - 1 & \text{if $x < y$,} \\
1 + y - 2 x & \text{if $x > y$.}
\end{cases}
\end{equation*}
Using Theorem \ref{eq}, the $x$ player should earn the expected payoff of $0$ when the $y$ player uses $F$.
That is, 
\begin{equation}
\label{check}
\begin{split} 
0 &= \int_{\mathbb{R}} A(x,y)  F'(y) \, dy \\
&= \int_0^x (1 + y - 2 x) F'(y) \, dy + \int_x^u (2 y - x - 1) F'(y) \, dy.
\end{split}
\end{equation}
At this point, we can simply verify that the function $F$ given in the statement of the theorem 
satisfies the above equation, but instead we will finish the proof by showing how such an $F$ is 
obtained.

Differentiating both sides of (\ref{check}) to eliminate the integrals, simplifying, 
and using $F(u) = 1$ and $F(0) = 0$, we 
arrive at the differential equation
\begin{equation*}
0 = 2(x-1) F'(x) + 1 + F(x).
\end{equation*}
The solution is $F(x) = C(2- 2 x)^{-1/2} - 1$ on $(0,u)$ for some constants $C$ and $u$.  Using
$F(0) = 0$, it follows that $C = \sqrt{2}$.  Using the fact that $F(x) \leq 1$ for all $x$, $u = 3/4$.
We have now found the function $F$ given in the statement of the theorem. 
\end{proof}

The probability that both players guess too high can be found using Theorem \ref{2}.  To find this 
probability, we first note that the probability that both players guess too high, provided $r$ is 
a random number selected uniformly, is $(1 - F(r))^2$ where $F(x)$ is the function in Theorem 
\ref{2}.  Therefore, the probability that both players guess too high is  
\begin{equation*}
\int_0^{3/4} (1 - F(r))^2 \, dr 
= \int_0^{3/4} \left(2 - \frac{1}{\sqrt{1-r}} \right)^2 \, dr 
= \ln 4 - 1
\approx 0.3863.
\end{equation*}

We now move on to the $n$ person game for $n \geq 3$.  The theory of continuous $n$ person symmetric 
zero-sum games is significantly less developed than that for two players, but we still are able to 
generalize the approach taken in the two person game. 

In the two person game, no player should guess a number larger than $3/4$.  What is the least upper 
bound for a player's guess for the $n$ person game? 

\begin{theorem}
\label{lub}
If there is an optimal strategy given by a differentiable cumulative distribution function $F$, 
then the least upper bound for a player's guess in the $n$-person game is 
$1 - \frac{1}{n} + \frac{1}{n^2}$ for $n \geq 2$. 
\end{theorem}

\begin{proof}
Let $A(x_1,\dots,x_n)$ be the expected payoff to the first player when player $i$ guesses $x_i$ and 
let $u$ be the least upper bound for the set of possible guesses.  Using Theorem \ref{eq}, the 
function $F$ should satisfy 
\begin{equation}
\label{iint}
0 = \int \cdots \int A(x_1,\dots,x_n) F'(x_2) \cdots F'(x_n) \, dx_n \cdots dx_2
\end{equation}
for all $x_1$ where the integral is over the region $[0,u] \times \cdots \times [0,u]$.

In the case where $x_1 = 0$ and over the region where $x_2 \leq x_3 \leq \cdots \leq x_n$, the 
integral in \eqref{iint} becomes
\begin{multline}
\label{iint2}
\int_{0}^u \int_{x_2}^u \cdots \int_{x_{n-1}}^u  (n x_2 - 1) F'(x_2) 
\cdots F'(x_n) \, dx_n \cdots dx_2 \\
=  \int_{0}^u \frac{n x_2 - 1}{(n-2)!} (1 - F(x_2))^{n-2} F'(x_2) \, dx_2.
\end{multline}
To evaluate the $n-2$ integrals above, we used the fact that $F(u) =1$.  If the variables 
$x_2,\dots,x_n$ are written in any order other than $x_2 \leq x_3 \leq \cdots \leq x_n$, then the 
integral in \eqref{iint2} remains unchanged; this is because 
$A(0,x_2,\dots,x_n) = n \min \{x_2,\dots,x_n\} - 1$ and because we are able to interchange the order 
of integration.  There are $(n-1)!$ orders of the variables $x_2,\dots,x_n$.  Therefore, when 
$x_1 = 0$, \eqref{iint} becomes
\begin{equation*}
0 = (n-1) \int_0^u (n x_2 - 1) \left(1 - F(x_2) \right)^{n-2} F'(x_2) \, dx_2.
\end{equation*}
This equation can be rewritten as
\begin{equation}
\label{help1}
\int_0^u x_2  \left(1 - F(x_2) \right)^{n-2} F'(x_2) \, dx_2 = \frac{1}{n(n-1)}.
\end{equation}

At the other extreme, consider $x_1 = u$.  Assuming $x_2 \leq \cdots \leq x_n$, the integral in 
\eqref{iint} becomes 
\begin{multline*}
\int_{0}^u \int_{x_2}^u \cdots \int_{x_{n-1}}^u  (n - 1 - n u + x_2) F'(x_2) 
\cdots F'(x_n) \, dx_n \cdots dx_2 \\
=  \int_{0}^u \frac{ (n - 1 - n u + x_2)}{(n-2)!} (1 - F(x_2))^{n-2} F'(x_2) \, dx_2.
\end{multline*}
The integral above is invariant for any of the $(n-1)!$ orderings of the variables $x_2,\dots,x_n$.  So, if 
$x_1 = u$, equation \eqref{iint} becomes
\begin{equation*}
0 = (n-1) \int_0^u  (n - 1 - n u + x_2) \left(1 - F(x_2) \right)^{n-2} F'(x_2) \, dx_2.
\end{equation*}
Using equation \eqref{help1}, the above equation simplifies to $0 = (n - 1 - n u) + \frac{1}{n}$. 
Solving for $u$ proves the theorem. 
\end{proof}

\begin{theorem} 
\label{3}
An optimal strategy for the three person ``Price Is Right'' guessing game is approximately the 
cumulative distribution function 
\begin{align*}
F(x)  = 
&  0.726193 x+0.480269 x^2+0.181628 x^3+0.0444137 x^4+0.0509559 x^5 \\ 
&+0.0648413 x^6+0.0363864 x^7 + 0.0123602 x^8+0.0178371 x^9 \\ 
&+0.0244243 x^{10}+0.0144697 x^{11}+0.00535718 x^{12}+
\cdots 
\end{align*}
for $x \in [0,7/9]$.  See Figure 1 for the graph of this function.
\begin{figure}
\begin{center}
\includegraphics{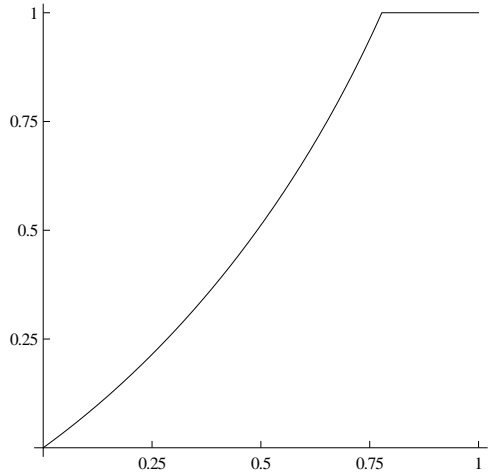}
\end{center}
\caption{An optimal strategy for the three player ``The Price Is Right'' guessing game.}
\end{figure}
\end{theorem}

\begin{proof} 
Let $A(x,y,z)$ be the expected payoff to the first player when he guesses $x$ and his opponents guess $y$ and 
$z$. Assuming that $y < z$, we have 
\begin{equation*}
A(x,y,z) 
= 
\begin{cases}
3 y - 2 x - 1 & \text{if $x < y < z$,} \\
3 z + y - 3 x - 1 & \text{if $y < x < z$,} \\
2 + y - 3 x & \text{if $y < z < x$.} 
\end{cases}
\end{equation*}
For the other three orders, we use the symmetry $A(x,y,z)=A(x,z,y)$. We also disregard the possibility  
that $x=y$, $y=z$, or $x=z$ with the expectation that we find an optimal cumulative distribution function $F$ 
that does not give positive probability to single points.

The function $F$ should satisfy 
\begin{equation*}
0 = \iint A(x,y,z) F'(y) F'(z) \, dy \, dz  
\end{equation*}
where the integral is taken over $[0,7/9] \times [0,7/9]$ (this $7/9$ comes from Theorem \ref{lub}).
Breaking this integral into six pieces, expanding, and simplifying where possible using $F(0) = 0$ 
and $F(7/9) = 1$, we find
\begin{multline*}
0 = -1-2 x-2 x F(x)+3 F(x)^2+x F(x)^2+2 \int_0^x y F'(y) \, dy+6 \int_x^{7/9} y F'(y) \, dy \\
+ 6 F(x) \int_x^{7/9} y F'(y) \, dy-2 \int_0^x y F(y) F'(y) \, dy-6 \int_x^{7/9} y F(y) F'(y) \, dy.
\end{multline*}
Differentiating, we obtain
\begin{equation}
\label{d1}
\begin{split}
0 = -2&-2 F(x)+F(x)^2-6 x F'(x)+6 F(x) F'(x) \\ &+6 \left(\int_x^{7/9} y F'(y) \, dy\right) F'(x).
\end{split}
\end{equation}
The term with the integral sign is the only term which prohibits us from finding an ordinary 
differential equation.  Differentiating \eqref{d1} once more gives 
\begin{equation}
\label{d2} 
\begin{split}
0 = -8 & F'(x)+2 F(x) F'(x)+6 F'(x)^2-6 x F'(x)^2-6 x F''(x) \\ & +6 F(x) F''(x)
+6 \left(\int_x^{7/9} y F'(y) \, dy\right) F''(x).
\end{split}
\end{equation}
Again, the term with the integral sign is the only term which prohibits us from finding an ordinary 
differential equation, but now we can use \eqref{d2} to find an expression for 
$\int_x^{7/9} y F'(y) \, dy$ and use it to replace this integral in \eqref{d1}.  Doing so yields 
\begin{equation*}
0 = \frac{8 F'(x)^2-2 F(x) F'(x)^2-6 F'(x)^3+6 x F'(x)^3-2 F''(x)-2 F(x) F''(x)+F(x)^2 F''(x)}
{F''(x)} .
\end{equation*}
It is possible to to use the numerator of the above differential equation, with the help of the 
boundary conditions $F(0) = 0$ and $F(7/9) = 1$, to find a power series solution for $F(x)$.  
This was done by plugging a function of the form $F(x) = \sum c_i x^i$ for some coefficients $c_i$ into the 
differential equation and finding the coefficients $c_i$ recursively.  This process finds the power series in 
the statement of the theorem. 
\end{proof}

Solutions to the $n$ player game with $n \geq 4$ can be found generalizing the proof of Theorem 
\ref{3}, although working out the details may be unreasonable.  Solving the four person game, for 
instance, involves repeatedly differentiating an integro-differential equation in order to 
back-substitute for unknown terms similar to $\int_0^{7/9} y F'(y) \, dy$.  Working through these 
calculations to find a numerical approximation for $F$ with acceptable accuracy can be a challenge 
for a computer algebra system running on a standard desktop computer.

\section{The Closest Wins.} 
\label{closest}

In this section, we study the payoff scheme so that the winner of the guessing game is the player with the 
guess closest to the random number $r$.  If the number $r$ is selected uniformly from the interval $[0,1]$ and 
$k \geq 1$ of the $n$ players select the same winning number, then the payoff is $n/k - 1$ to each 
winner and $-1$ to everyone else.  This payoff is the result of each player anteing $1$ with winning players 
sharing the pot.

The two person game has a saddle point: each player should always guess $1/2$.  For three or more players, 
however, guessing $1/2$ every time is not optimal---if a player knew that each of his opponents were playing 
$1/2$, then it would be better to deviate from this strategy by playing $1/2+\varepsilon$ for a small 
enough $\varepsilon$.

For the three person game, a first strategy to check is the uniform distribution on an interval 
centered around $1/2$.  

\begin{theorem}
\label{3person}
The uniform distribution on the interval $[1/4,3/4]$ is optimal for the three person game.
\end{theorem}

\begin{proof}
Let $A(x,y,z)$ be the expected payoff to the first player when he guesses $x$ and his opponents guess $y$ and 
$z$. Assuming $y < z$, we have
\begin{equation*}
A(x,y,z) 
= 
\begin{cases}
3(x+y)/2 - 1 & \text{if $x < y < z$,} \\
3(z-y)/2 - 1 & \text{if $y < x < z$,} \\
2 - 3(x+z)/2 & \text{if $y < z < x$.}
\end{cases}
\end{equation*}
For the other three orders, we use the symmetry $A(x,y,z)=A(x,z,y)$. 

Routine integration shows that $F(x)=2(x-1/2)$ satisfies the equation
\begin{equation*}
0 = \int_{1/4}^{3/4} \int_{1/4}^{3/4}A(x,y,z) F'(y) F'(z) \, dy \, dz ,
\end{equation*}
as needed to prove the theorem.
\end{proof}

It is possible to use the symmetry of the game about $1/2$ to find the solution in Theorem \ref{3person}. 
This approach is explained and used in Theorem \ref{4person} to find a solution to the four person game.

\begin{theorem}
\label{4person} 
A solution to the four person guessing game is for each player to use the cumulative distribution
function given approximately by 
\begin{align*}
F(x)  = &\frac{1}{2} + 0.636459 {(2 x - 1)} + 0.214848 {(2 x - 1)}^3 + 0.144471 {(2 x - 1)}^5 \\
&+ 0.123155 {(2 x - 1)}^7 
 +  0.118372 {(2 x - 1)}^9 + 0.122339 {(2 x - 1)}^{11} \\ &+ 0.132742 {(2 x - 1)}^{13} 
+ 0.149142 {(2 x - 1)}^{15} + \cdots
\end{align*}
on the interval $[0.174989,0.825011]$.  

See Figure 2 for the graph of this function.
\begin{figure}
\begin{center}
\includegraphics{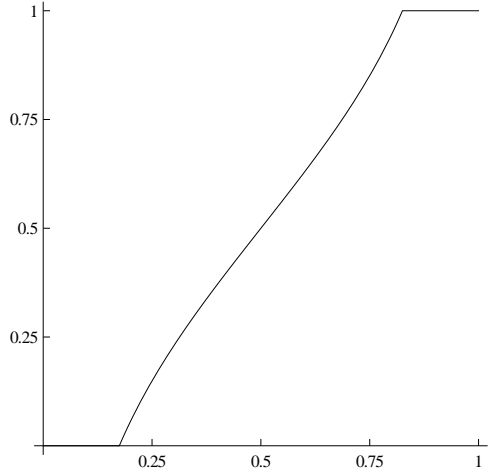}
\end{center}
\caption{An optimal strategy for the four player ``Closest Wins'' guessing game.}
\end{figure}
\end{theorem}

\begin{proof}
Let $A(x,y,z,w)$ be the function giving the expected payoff when players select $x,y,z,w \in [0,1]$.
Let $F$ be a cumulative distribution function giving an optimal strategy for each player in the game 
$A$.  The symmetry of the game $A$ implies that $F(1/2 - x/2) = 1 - F(1/2 + x/2)$ for all 
$x \in [0,1]$; this is simply saying that the probability that a player selects a number less than 
$1/2 - x/2$ should be the same as the probability that a number greater than $1/2 + x/2$. 

Instead of analyzing $A$, we will consider a related game $\widehat{A}$ by taking the symmetry about 
$1/2$ into account.  Let $\widehat{A}(x,y,z,w)$ be the game defined by
\begin{equation*}
\widehat{A}(x,y,z,w) 
= \frac{1}{16} \sum A \left( \frac{1 \pm x}{2}, \frac{1 \pm y}{2}, \frac{1 \pm z}{2}, 
\frac{1 \pm w}{2} \right)
\end{equation*}
where the sum runs over all 16 independent choices for the $+$ or $-$ in each argument.  It is useful
to think of $\widehat{A}$ as the expected payoff for the following game.  Four players select 
$x,y,z,w \in [0,1]$.  Then, $1/2+x/2$ and $1/2-x/2$ will be played in $A$ with probability $1/2$ 
each, $1/2+y/2$ and $1/2-y/2$ will be played in $A$ with probability $1/2$ each, and similarly for 
$z$ and $w$.  After calculating the expected payoffs, $\widehat{A}$ satisfies
\begin{equation}
\label{wideA}
\widehat{A}(x,y,z,w) =
\begin{cases}
(-8 + 16 y + 8 z + 4 w)/16 & \text{if $x < y < z < w$},  \\
(-8 - 4 x + 8 z + 4 w)/16  & \text{if $y < x <  z < w$}, \\
(-8 x - 8 z + 8w)/16 & \text{if $y < z < x < w$}, \\
(-8 + 16 x - 4 z - 8 w)/16 & \text{if $y < z < w < x$}.
\end{cases}
\end{equation}
The other values of $\widehat{A}(x,y,z,w)$ can be found using symmetry. 

If $\widehat{F}$ is a cumulative distribution function giving an optimal solution for $\widehat{A}$, then  
\begin{align*}
\widehat{F}(x) 
& = (\text{the probability of playing $\leq x$ in the game $\widehat{A}$}) \\
& = (\text{the probability of playing in  $(1/2 - x/2, 1/2 + x/2)$ in the game $A$}) \\
& = F(1/2 + x/2) - F(1/2 - x/2) \\
& = 1 - 2 F(1/2 - x/2)
\end{align*}
where the last line used the identity $F(1/2 + x/2) = 1 - F(1/2 - x/2)$.  Solving for $F(x)$ gives
\begin{equation*} 
F(x) = 
\begin{cases}
1/2 - \widehat{F}(1 - 2 x)/2   & \text{if $x \in [0,1/2]$} \\
1/2 + \widehat{F}(2 x - 1)/2 & \text{if $x \in [1/2,1]$}.
\end{cases}
\end{equation*}
If $\widehat{F}$ happens to be an odd function, then
\begin{equation}
\label{odd} 
F(x) = 1/2 + \widehat{F}(2 x - 1)/2.
\end{equation}

The function $\widehat{F}$ should satisfy
\begin{equation*}
0 = \iiint \widehat{A}(x,y,z,w) \widehat{F}'(y) \widehat{F}'(z) \widehat{F}'(w) \, dy \, dw \, dz 
\end{equation*}
where the integral is over $[0,u] \times [0,u] \times [0,u]$.  Evaluating the integral by splitting 
it up into 24 regions and using our knowledge of $\widehat{A}$ displayed in \eqref{wideA}, we arrive 
at the following equation:
\begin{equation}
\label{first}
\begin{split}
0 = 
&-\frac{1}{2}-\frac{3}{4} x \widehat{F}(x)+\frac{3 \widehat{F}(x)^2}{2}-\frac{1}{4} x \widehat{F}(x)^3
+3 \int_x^u y \widehat{F}'(y) \, dy \\ 
&+\frac{3}{4} \widehat{F}(x)^2 \int_x^u y \widehat{F}'(y) \, dy 
-3 \int_0^x y \widehat{F}(y) \widehat{F}'(y) \, dy \\
&+\frac{3}{2} \widehat{F}(x) \int_0^x y \widehat{F}(y) \widehat{F}'(y) \, dy
-3 \int_x^u y \widehat{F}(y) \widehat{F}'(y) \, dy \\
&+\frac{3}{4} \int_x^u y \widehat{F}(y)^2 \widehat{F}'(y) \, dy.
\end{split}
\end{equation}
Differentiating, we obtain
\begin{equation}
\label{annoying}
\begin{split}
0 =
&-\frac{3 \widehat{F}(x)}{4}-\frac{\widehat{F}(x)^3}{4}-\frac{15}{4} x \widehat{F}'(x)
+3 \widehat{F}(x) \widehat{F}'(x)-\frac{3}{4} x \widehat{F}(x)^2 \widehat{F}'(x) \\
&+\frac{3}{2} \widehat{F}(x) \left(\int_x^u y \widehat{F}'(y) \, dy\right) \widehat{F}'(x)
+\frac{3}{2} \left(\int_0^x y \widehat{F}(y) \widehat{F}'(y) \, dy\right) \widehat{F}'(x).
\end{split}
\end{equation}
The integral terms in the above equation are annoying.  To get rid of them, 
differentiate \eqref{annoying}, solve for $\int_x^u y \widehat{F}'(y) \, dy$, and use the result to 
replace the $\int_x^u y \widehat{F}'(y) \, dy$ terms in \eqref{annoying}.  Then take that equation, 
differentiate it, solve for $\int_0^x y \widehat{F}(y) \widehat{F}'(y) \, dy$, and replace the
$\int_0^x y \widehat{F}(y) \widehat{F}'(y) \, dy$ terms.  After doing this and simplifying we can 
find the differential equation 
\[
\begin{split}
0 = 
&-18 \widehat{F}(x) \widehat{F}'(x)^4-12 x \widehat{F}'(x)^5+21 \widehat{F}'(x)^2 \widehat{F}''(x)
+9 \widehat{F}(x)^2 \widehat{F}'(x)^2 \widehat{F}''(x) \\
& -9 \widehat{F}(x) \widehat{F}''(x)^2 
-3 \widehat{F}(x)^3 \widehat{F}''(x)^2+3 \widehat{F}(x) \widehat{F}'(x) \widehat{F}'''(x)
+\widehat{F}(x)^3 \widehat{F}'(x) \widehat{F}'''(x).
\end{split}
\]
It is possible to find a power series solution for $\widehat{F}$ using this differential equation.
When this is done, we find 
\begin{equation}
\label{power}
\begin{split}
\widehat{F}(x) = &a x+0.208333 a^3 x^3+0.0864583 a^5 x^5+0.0454861 a^7 x^7 \\
&+0.026982 a^9 x^9+0.0172103 a^{11} x^{11} +0.0115247 a^{13} x^{13} \\
&+0.00799137 a^{15} x^{15}+0.0056886 a^{17} x^{17} + \cdots
\end{split}
\end{equation}
where $a = \widehat{F}'(0)$.  This function is odd. 

Now we only need to find the unknown constants $a$ and $u$.  
Since $\widehat{F}'(0) \geq 0$, we know $a = \widehat{F}'(0) > 0$ or else $\widehat{F}$ would be the 
zero function.  All coefficients in $\widehat{F}$ are positive, so there is exactly one real root to the 
equation $\widehat{F}(x/a) = 1$; it is approximately $0.82742$.  The constant $u$ satisfies  
$\widehat{F}(u) = 1$, giving $u = 0.82742/a$.  Now we have a relationship between $u$ and $a$.

Taking $x = u$ in \eqref{first} yields
\begin{equation*}
0 = 1 - u - \frac{3}{2} \int_0^u y \widehat{F}(y) \widehat{F}'(y) \, dy.
\end{equation*}
Substituting \eqref{power} and $u = 0.82742/a$ into this expression, we find 
$0 = 1 - 1.27292/a$, implying that $a \approx 1.27292$.  This gives $u \approx .650021$. 
Putting everything together (using \eqref{odd} to simplify since $\widehat{F}$ is odd), we find the 
function $F$ given in the statement of the theorem. 
\end{proof}

The method used in Theorem \ref{4person} can be used to find the solution in Theorem 
\ref{3person}.  Hypothetically, the solution for $n$ person games for $n \geq 5$ can be worked out 
in a similar manner, but the calculations may prove to be unreasonable.  

\section{Numerical Approximations.}
\label{approx}

Since our derivations of the optimal solutions in Sections \ref{pir} and \ref{closest} were sufficiently 
complicated, we found it reassuring to find numerical support for our results.  In this section we describe an 
algorithm used to find an approximate solution for a discrete analogue of the three person ``Price Is Right'' 
guessing game.  Similar algorithms will work for the more players and the for the ``Closest Wins" guessing 
game.  Our idea is similar to that of fictitious play \cite{julia}, but we will not attempt to prove that our 
approximations converge to an optimal strategy.  

We make the game finite by letting the target number be uniformly chosen among the integers from $1$ to $N$ 
where $N$ reasonably large.  A mixed strategy is a probability vector $\mathbf{p} = (p_1,\ldots,p_N)$ where
$p_i$ gives the probability of selecting the number $i$.  

Let $a(i,j,k)$ be the payoff to the first player when the three guesses are $i,j$ and $k$. 
Let $\mathbf{p}$ be the current candidate for an optimal solution.  When player 1 guesses $i$ and 
the other players both use the mixed strategy $\mathbf{p}$, then the payoff to player 1, denoted by 
$v_i(\mathbf{p})$, is given by
\begin{equation*}
v_i(\mathbf{p}) = \sum_{j,k} a(i,j,k)p_j p_k.
\end{equation*}

If $\mathbf{p}$ is an optimal solution, then $v_i(\mathbf{p})=0$ for each $i$ such that $p_i > 0$. 
However, if $\mathbf{p}$ is not optimal, then $v_i(\mathbf{p}) > 0$ for some $i$, and so it would be 
advantageous to the other players to increase the probability that they guess those values $i$. 
So we increase $p_i$ by a quantity that is proportional to $v_i(\mathbf{p})$.  This is our new candidate for an
optimal solution.

Let $\varepsilon >0$ be a small constant of proportionality.  From $\mathbf{p}$ we compute the 
vector $\mathbf{r} = (r_1,\ldots,r_N)$ with $r_i = p_i + \varepsilon \max(0,v_i(\mathbf{p}))$. The 
components of $r$ are non-negative, but $r$ may not be a probability vector and so we rescale it be 
$\mathbf{p}' = \mathbf{r}/(r_1 + \cdots + r_n)$. 

For an initial $\mathbf{p}$ we use the uniform probability $p_i = 1/N$.  We found that if $\varepsilon$ is 
too large, then the plot of $p$ does not settle down, and if $\varepsilon$ is too small, then convergence 
is too slow. For the three player game and $N=50$, we found that with $\varepsilon = 0.001$ and running $5000$ 
iterations, we arrived at a very good approximation to the exact solution, whose graph is shown in Figure 3.  It nicely approximates the solution displayed in Figure 1. 
\begin{figure}[ht]
\begin{center}
\includegraphics{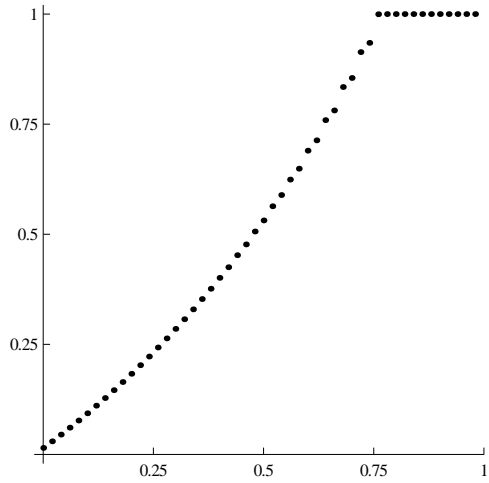}
\end{center}
\caption{An approximation of an optimal strategy for three players with $N=50$ and $5000$ iterations.}
\end{figure}

\bigskip

\noindent\textit{Mathematics Department, 
Cal Poly, San Luis Obispo, CA 93422 \\
aamendes@calpoly.edu}

\bigskip

\noindent\textit{American Institute of Mathematics, 360 Portage Avenue, Palo Alto, CA 94306 \\
kmorriso@calpoly.edu}

\end{document}